\documentclass[conference]{IEEEtran}
\usepackage{amsmath, amssymb, epsfig, verbatim}

\newtheorem{theorem}{Theorem}

\newtheorem{prop}[theorem]{Proposition}
\newtheorem{corollary}[theorem]{Corollary}
\newtheorem{definition}{Definition}
\newtheorem{example}{Example}

\def\ds{\displaystyle}
\def\lt{\left}
\def\rt{\right}
\def\C{\mathbb{C}} %Complex%
\def\Z{\mathbb{Z}} %Integers%
\def\F{\mathbb{F}} %Finite Field%
\def\p{\mathcal{P}} %Set of Participants%
\def\c{\mathcal{C}} %Code%
\def\dc{\mathcal{C}^{\perp}} %Dual Code%
\def\as{\Gamma} %Access Structure%

\DeclareMathOperator{\supp}{supp}
\DeclareMathOperator{\wt}{wt}

\begin{document}

\title{An extension of Massey scheme for secret sharing}

\author{\IEEEauthorblockN{Romar dela Cruz\IEEEauthorrefmark{1}\IEEEauthorrefmark{2}, 
     Annika Meyer\IEEEauthorrefmark{3},
     and Patrick Sol\'e\IEEEauthorrefmark{4}}
   \IEEEauthorblockA{\IEEEauthorrefmark{1}%
     Division of Mathematical Sciences, SPMS, Nanyang Technological University, Singapore\\
     Email: roma0001@ntu.edu.sg}
   \IEEEauthorblockA{\IEEEauthorrefmark{2}%
     EDSTIC, Universit\'{e} de Nice-Sophia Antipolis, Les Algorithmes, Euclide B 06903 Sophia Antipolis, France}
   \IEEEauthorblockA{\IEEEauthorrefmark{3}%
     Lehrstuhl D f\"{u}r Mathematik, RWTH Aachen University, Templergraben 64, 52062 Aachen, Germany\\
     Email: annika.meyer@math.rwth-aachen.de}
   \IEEEauthorblockA{\IEEEauthorrefmark{4}%
     CNRS, Telecom-ParisTech, Dept Comelec, 46 rue Barrault 75013 Paris, France\\
     Email: sole@enst.fr}}

\maketitle

\begin{abstract}
We consider an extension of Massey's construction of secret sharing schemes using linear codes.  We describe the access structure of the scheme and show its connection to the dual code.  We use the $g$-fold joint weight enumerator and invariant theory to study the access structure.
\end{abstract}

%{\bf Keywords:} secret sharing, self-dual codes, joint weight enumerator

\bigskip
\section{Introduction}
A \emph{secret sharing scheme} is a process of distributing a secret to a set of participants in such a way that only certain subsets of them can determine the secret.  The set of all subsets which can determine the secret is called the access structure of the scheme.  Secret sharing schemes were introduced in 1979 (\cite{B79}, \cite{S79}) and since then, different schemes were constructed.  For a general introduction to secret sharing schemes, see for instance \cite{St92}.  An important class of secret sharing schemes are those which are based on linear codes.  The relation between secret sharing schemes and linear codes was first presented in \cite{MS81}.  The access structure of schemes based on self-dual codes was analyzed in \cite{DMS08} using some properties of the codes.   

In this work, we consider an extension of the construction method in \cite{M95}. This construction is presented in Section 2.  In Section 3, we characterize the groups that can determine the secret.  In Sections 4-6, we describe the access structure of the scheme by extending the techniques used in \cite{DMS08}. 

\section{Codes and Secret Sharing Schemes}
Let $\F_q$ stand for the finite field of order $q$, where $q$ is a prime power.  The \emph{Hamming weight}
$\wt(\vec{v})$ of a vector $\vec{v}$ in $\F_q^n$ is the number of its non-zero coordinates while the support of $\vec{v}$ is given by supp$(\vec{v})=\{i : v_i\neq 0, 1\leq i\leq n\}$.  An $[n,k,d]$ \emph{linear code} $\c$ is a linear subspace of $\F_q^n$ where $k$ is the dimension and $d$ is the minimum Hamming weight.  A \emph{generator matrix} $G$ for a code $\c$ is a matrix whose rows form a basis for $\c$.  For any linear code $\c$, we denote by $\dc$ its dual under the usual inner product.  A code $\c$ is said to be \emph{self-orthogonal} if $\c\subseteq \dc$ and it is \emph{self-dual} if $\c=\dc$.

We consider the following secret sharing scheme.  Let $\p=\{P_1,\ldots,P_n\}$ be the set of participants.  Suppose we want to share the secret $\vec{s}=(s_1, s_2, \ldots, s_l)\in\F_q^l$.  Let $\c$ be an $[l+n,k,d]$ linear code over $\F_q$ with $d>l$. Consider a generator matrix $G=[G_1,\ldots,G_l,G_{l+1},\ldots,G_{l+n}]$ of $\c$ where $G_i$ is the $i$th column.  To generate the shares, the dealer picks a vector $\vec{u}$ such that $\vec{u}G_i=s_i$ for $1\leq i\leq l$.  A codeword $\vec{c}=\vec{u}G$ is then computed.  Now the share of $P_i$ is $c_{l+i}$ for $i=1,\ldots,l$.  Note that when $l=1$ then we have Massey's construction \cite{M95}.  We also remark that this construction was mentioned in \cite{MS81} in the case of Reed-Solomon codes. 

Let $B=\{P_{i_1},\ldots,P_{i_m}\}\subseteq \p$.  We have the following result from \cite{BK93}.  The participants in $B$ can recover the secret if span$(G_1,\ldots,G_l)\subseteq$ span$(G_{i_1},\ldots,G_{i_m})$.  The participants in $B$ have no information on the secret if span$(G_1,\ldots,G_l)\cap$ span$(G_{i_1},\ldots,G_{i_m})=\vec{0}$.  Otherwise, the participants in $B$ have partial information on the secret.

The \emph{access structure} $\as$ of the scheme is the collection of all subsets of $\p$ that can recover the secret.  An element $B\in\as$ is called a \emph{minimal access group} if no element of $\as$ is a proper subset of $B$.  For $l=1$, it was shown in \cite{M95} that there is a one-to-one correspondence between the set of minimal access groups and the set of minimal codewords of $\dc$ with first coordinate equal to 1.

A scheme is said to be \emph{perfect} if every group in the access structure can determine the secret and every group not in the access structure has no information about the secret.  If a scheme is not perfect then some groups have partial information on the secret.  The scheme that we consider here is non-perfect for $l\geq 2$.  

The \emph{information rate} of a scheme is the ratio of the size of the secret and maximum size of the share.  For perfect schemes, the size of each share must be at least as large as the size of the secret.  An advantage of non-perfect schemes is that the size of each share can be smaller than the size of the secret.  The information rate of the scheme above is $l$.  

\section{Access Structure}
We now describe the access structure of a scheme based on a linear code $\c$.
In \cite{CCGHV07}, it was shown that any group of size at most $d^{\bot}-l-1$ has no information about the secret and any group of size at least $n+l-d+1$ can recover the secret.
Here we show that no group of size at most $d_l^{\perp}-l-1$ is in the access structure, where $d_l^{\perp}$
is the $l$th generalized Hamming weight of $\c$ (cf. Corollary \ref{lth_hamming_weight}).
Since $d_l^{\perp}$ is not so easy to determine for $l \ge 2$, we also show that the size of an access group is at least
$\frac{3}{2}(d^{\perp}-l)$, where $d^{\perp}$ is the minimum weight of $\c^{\perp}$ (cf. Corollary \ref{A_n_k}).
This bound is weaker than the one given by $d_l^{\perp}$, but easier to calculate.
We are going to use the following proposition which is an extension of the approach in \cite{M95}.

\begin{prop}\label{main_prop}
Let $B=\{P_{i_1},\ldots,P_{i_m}\}\subseteq \p$.  Then the participants in $B$ can determine $\vec{s}$ if and only if there exist codewords $\vec{v}_j\in\dc$, $1\leq j\leq l$, satisfying the following conditions:
\begin{itemize}
\item[i.] The subvector of $\vec{v}_j$ consisting of its first $l$ coordinates is equal to the $j$th unit vector $\vec{e}_j$ in $\F^l_q$.
\item[ii.] supp$(\vec{v}_j)\subseteq \{j,i_1,\ldots,i_m\}.$
\end{itemize}
\end{prop}
\begin{proof}
Suppose there exist codewords $\vec{v}_j\in\dc$, $1\leq j\leq l$, satisfying conditions (i) and (ii).  For $j=1,\ldots,l$, we have $$\vec{s}\cdot \vec{v}_j = c_j + \ds\sum_{r=1}^m\alpha_{jr}c_{i_r}=0$$ for some constants $\alpha_{jr}, 1\leq r\leq m$, which are not all zero.  Hence, the secret $\vec{s}$ can be determined as a linear combination of the shares of participants in $B$.

Suppose the participants in $B$ can determine the secret.  Then for each $j=1,\ldots,l$, we have an equation of the form $$c_j = \ds\sum_{r=1}^m \beta_{jr} c_{i_r}$$ for some constants $\beta_{jr}, 1\leq r\leq m$, which are not all zero.  The equation can be rewritten as
$$(c_1, c_2, \ldots, c_l, c_{l+1},\ldots, c_{l+n})\cdot$$
$$(\vec{e}_j,0,\ldots,-\beta_{j1},\ldots, -\beta_{jm},0,\ldots,0)=0.$$ Now the codewords $(\vec{e}_j,0,\ldots,-\beta_{j1},\ldots, -\beta_{jm},0,\ldots,0)$ are in $\dc$ and satisfy conditions (i) and (ii).  
\end{proof}

\begin{example}
Let $\c_1$ be the $[8,3,4]$ linear code over $\F_3$ with generator matrix 
$$G = \lt[\begin{array}{cccccccc}
1 & 0 & 0 & 0 & 2 & 2 & 1 & 1\\
0 & 1 & 0 & 1 & 2 & 1 & 2 & 1\\
0 & 0 & 1 & 2 & 0 & 1 & 0 & 2
\end{array}\rt].$$
We consider the scheme based on the dual of $\c_1$ with $l=2$ (so we have 6 participants).  Applying the proposition, we can verify that the access structure consists of 4 groups of size 5 and 1 group of size 6.
\end{example}

\begin{example}
Consider the scheme based on the $[8,4,4]$ extended binary Hamming code with $l=3$.  In this case, we have a total of 5 participants.  There are 4 groups of size 4 and 1 group of size 5 in the access structure.
\end{example}

\begin{corollary}\label{lth_hamming_weight}
Any group of $d_l^{\bot}-l-1$ or less participants is not in the access structure where $d_l^{\bot}$ is the $l$th generalized Hamming weight of $\dc$.
\end{corollary}

\begin{proof}  
The $l$th generalized Hamming weight of a linear code is the minimum support of its subcodes of dimension $l$.
A minimal access group $B=\{P_{i_1},\ldots,P_{i_m}\}$ corresponds to an $[l+n,l]$ subcode $\mathcal{D}$ of $\dc$ such that supp$(\mathcal{D})=\{1,\ldots,l,i_1,\ldots,i_m\}$.  Hence, $m\geq d_l^{\bot}-l$.
\end{proof}

\begin{corollary}\label{A_n_k}
 If $l \ge 2$ then any group of $\frac{3}{2}(d^{\perp}-l)-1$ or less participants is not in the access structure.
\end{corollary}

\begin{proof}
As in the proof of Corollary \ref{lth_hamming_weight}, a minimal access group of size $m$ corresponds to an $[l+n,l]$ subcode $\mathcal{D}$ of $\dc$ whose support has size $l+m$. Moreover, deleting the first $l$ coordinates of $\mathcal{D}$ as well as those coordinates which are not in its support yields a binary $[m,l]$ code of minimum weight at least $d^{\perp}-l$.  Recall that $A(N,\delta)$ is the maximum size of a (not necessarily linear) code of length $N$ and minimum weight at least $\delta$.  The above yields $A(m,d^{\perp}-l) \ge 2^l >2$.  On the other hand, it is well-known that $A(N,\delta) \le 2$ whenever $N \le \frac{3}{2}\delta -1$.  This yields $m \ge \frac{3}{2}(d^{\perp}-l)$.
\end{proof}

\begin{prop}
When all participants come together and attempt to determine the secret, $\left\lfloor \frac{d-l}{2}\right\rfloor$ cheaters can be detected.
\end{prop}
\begin{proof}
Deleting the first $l$ coordinates of $\c$ results in a code with minimum distance $d-l$.
\end{proof}

\section{$g$-fold joint weight enumerator}
We describe the connection between the $g$-fold joint weight enumerator and the access structure.  The $g$-fold joint weight enumerator is a generalization of the joint weight enumerator (see \cite{DHO01}).

\begin{definition}
Let $A_1, A_2,\ldots,A_g$ be codes of length $n$ over $\F_q$.  The $g$-fold joint weight enumerator of $A_1, A_2,\ldots,A_g$ is defined as follows: 
\begin{align*}
&\mathcal{J}_{A_1, A_2,\ldots,A_g}(x_a;a\in\F_2^g)\\
&= \ds\sum_{\vec{c}_1\in A_1, \ldots, \vec{c}_g\in A_g} \ds\prod_{a\in\F_2^g} x_a^{n_a(\vec{c}_1,\ldots,\vec{c}_g)},
\end{align*}
where $\vec{c}_j=(c_{j1},\ldots,c_{jn})$, $n_a(\vec{c}_1,\ldots,\vec{c}_g) = |\{i|a=(\overline{c_{1i}},\ldots,\overline{c_{gi}})\}|$, and $\overline{c_{ji}}=1$ if $c_{ji}\neq 0$ and $\overline{c_{ji}}=0$ if $c_{ji}=0$.  Here $(x_a;a\in\F_2^g)$ is a $2^g$-tuple of variables with $\F_2^g$, that is, $(x_{00\ldots0}, x_{00\ldots1},\ldots, x_{11\ldots1})$. 
\end{definition}

First we consider the case $l=2$, i.e. the secret $\vec{s}=(s_1,s_2)$.  For simplicity, we use the corresponding decimal representation of the subscripts of the variables in the $g$-fold joint weight enumerator.  Let $T_1=\{1\}$ and $T_2=\{2\}$ with indicator vectors $1_{T_1}$ and $1_{T_2}$ respectively.  Consider the 4-fold joint weight enumerator $\mathcal{J}_{1_{T_1},1_{T_2},\dc,\dc}(x_a)$ where $a\in\F_2^4$.  We are interested in the coefficient $x_{10}x_5$. The coefficient is a polynomial in $x_0x_1x_2x_3$ and it gives information on the number and supports of pairs of codewords $\vec{u}, \vec{v}\in\dc$ whose first two coordinates are $(u_1,0)$ and $(0,v_2)$ respectively, where $u_1$ and $v_2$ are both non-zero.  In general, for secrets of length $l$ we use the $2l$-fold joint weight enumerator $\mathcal{J}_{1_{T_1},\ldots,1_{T_l},\dc,\ldots,\dc}(x_a;a\in\F_2^{2l})$ where $a\in\F_2^{2l}$. The following theorem
generalizes a result in \cite{DMS08} where Jacobi polynomials were used.

\begin{theorem}\label{main_theorem}
Let $X_1$ be the subset of $\F_2^{2l}$ consisting of all vectors whose first $l$ coordinates are zero and let
$X_2:=\{(\vec{e}_j,\vec{e}_j) \;|\; j \in \{1,\dots,l\} \}$, where $\vec{e}_j \in \F_2^l$ is the $j$th unit vector.
Then the coefficient of $ \prod_{a\in X_2} x_a$ in $\mathcal{J}_{1_{T_1},\ldots,1_{T_l},\dc,\ldots,\dc}(x_a;a\in\F_2^{2l})$
is a polynomial $p(x_a;a\in X_1)$. Identify $X_1$ with $\{0,\dots,2^l-1\}$ via the binary number representation and write
$$
  p=\sum_{ \mu \in \mathbb{N}_0^{2^l}}
    c_{\mu} \prod_{a \in X_1} x_a^{\mu_a}.
$$
Then the number $M_{\c}(m)$ of groups of size $m$ in the access structure of the scheme based on $\c$ satisfies
$$
  M_{\c}(m) \le \sum_{\mu} c_{\mu},
$$
where the sum is over all $\mu$ with $\sum_{i=1}^{2^l-1} \mu_i = m$
Moreover, if $m < \frac{3}{2}d^{\perp} -1$ then equality holds.
\end{theorem}

\begin{proof}
 The sum of the coefficients $c_{\mu}$, where $\sum_{i=1}^{2^l-1} \mu_i =m$, equals the number of tuples $(\vec{v}_1,\dots,\vec{v}_l)$ of
 elements of $\c^{\perp}$ such that the projection of $\vec{v}_j$ onto the first $l$ coordinates is the $j$th unit vector in $\F_2^l$,
 and
 $$
   |\cup_{j=1}^l \supp(\vec{v}_j) \cap \{l+1,\dots,l+n\}| = m.
 $$
 Hence due to Proposition \ref{main_prop}, every such tuple determines a group in the access structure of the scheme based on
 $\c$, and every minimal access group occurs as a union of supports of such a tuple.  However, in general there may be different tuples of codewords that correspond to the same access group.
 In this situation, there exists a tuple $(\vec{v}_1,\dots,\vec{v}_l)$ as above and an element $\vec{c} \in \c^{\perp}$ such that
 $$
   \supp(\vec{c}) \subseteq \cup_{j=1}^l \supp(\vec{v}_j) \cap \{l+1,\dots,l+n\}.
 $$
 Then for any $j \in \{1,\dots,l\}$, $|\supp(\vec{c}) \cap \supp(\vec{v}_j) \cap \{l+1,\dots,l+n\}| \ge \wt(\vec{c})+\wt(\vec{v}_j)-1-m$ and hence
 \begin{align*}
  d^{\perp}     &\le \wt(\vec{c}+\vec{v}_j)\\
                &\le  1+m-(\wt(\vec{c})+\wt(\vec{v}_j)-1 -m )\\
                &\le 2m+2-2d^{\perp},
 \end{align*}
 which yields $m \ge \frac{3}{2}d^{\perp} -1$. Hence if $m < \frac{3}{2}d^{\perp} -1$ then
 the sum of the coefficients $c_{\mu}$ with $\sum_{i=1}^{2^l-1} \mu_i =m$ equals the number of access groups of size $m$.
\end{proof}

\medskip

If $\c$ is self-orthogonal then there exists a weaker condition 
than the one in Theorem \ref{main_theorem} under which the number of access groups of size $m$
can be read off from the $2l$-fold joint weight enumerator.
To state this condition, we need the notion of the {\em code extension enumerator} below.

\begin{definition}
 Let $D$ be a linear self-orthogonal $[N,k,d]$ code. The {\em code extension enumerator}
 is the complex polynomial
 $$
  P_D(t)=\sum_{c} t^{d(\langle c,D \rangle)},
 $$
 where the sum is over a system of representatives of $D^{\perp}/D$.
\end{definition}

Clearly $\deg(P_D) \le d$, and a summand $t^{d'}$ in $P_D(t)$ gives rise to a linear
self-orthogonal $[N,k+1,d']$ code.

Now consider a secret sharing scheme based on a binary self-orthogonal linear code $\c$  and let $(\vec{v}_1,\dots,\vec{v}_l)$ be a tuple of
elements of $\c^{\perp}$ giving rise to an access group of size $m$, as in Proposition \ref{main_prop}.
Let $\mathcal{D}$ be the linear code generated by the $\vec{v}_j$, where the columns where all the $\vec{v}_j$ are zero are deleted.
Then $\mathcal{D}$ is a self-orthogonal $[l+m,l]$ code of minimum distance at least $d^{\perp}$.

Assume that there exists another tuple of elements of $\c^{\perp}$ leading to the same access group,
i.e. in Theorem \ref{main_theorem}, we have strict inequality for $M_{\c}(m)$.
Then there exists a nonzero element $\vec{c} \in \c^{\perp}$
with $\supp(\vec{c}) \subseteq \cup_{j=1}^l \supp(\vec{v}_j) \cap \{l+1,\dots,l+n\}$.
Let $(\vec{c})' \in \F_2^{l+m}$ be obtained from $\vec{c}$ by deleting
the coordinates where all the $\vec{v}_i$ are zero.
Then $\langle (\vec{c})',\mathcal{D} \rangle$ has minimum weight at least $d^{\perp}$,
hence gives rise to a summand $t^{d(\langle (\vec{c})',\mathcal{D}\rangle)}$ in $P_{\mathcal{D}}(t)$,
where $d(\langle (\vec{c})',\mathcal{D}\rangle) \ge d^{\perp}$. This yields

\begin{corollary}\label{cond_orth}
  Consider a secret sharing scheme based on a self-orthogonal linear code $\c$
 and let $\mathcal{T}$ be the set of all tuples in $\c^{\perp}$ that give rise to
 an access group of size $m$ (cf. Proposition \ref{main_prop}).
 For a tuple $(\vec{v_1},\dots,\vec{v_l}) \in \mathcal{T}$, let $\mathcal{D}(\vec{v_1},\dots,\vec{v_l})$
 be the code generated by the $\vec{v_j}$, in which the columns where all the $\vec{v_j}$ are zero are deleted.
  If for all such tuples, all monomials in $P_{\mathcal{D}(\vec{v}_1,\dots,\vec{v}_l)}(t)$
  (except for the monomial corresponding to $0 \in \mathcal{D}^{\perp}/\mathcal{D}$) have degree
  less than $d^{\perp}$ then equality holds in Theorem \ref{main_theorem}, i.e. the number
  of groups of size $m$ in the access structure of the scheme based on $\c$ can be read off from
  $\mathcal{J}_{1_{T_1},\ldots,1_{T_l},\dc,\ldots,\dc}$.
\end{corollary}

\section{Binary self-dual codes}
In this section, we focus on schemes based on binary self-dual codes and the case $l=2$.  Based on the previous section, we use $\mathcal{J}_{1_{T_1},1_{T_2},\c,\c}(x_0,\ldots,x_{15})$ and determine the coefficient of $x_{10}x_5$.  Let us denote this coefficient by $Z$.  Under some conditions, we can determine $Z$ using the biweight enumerator of $\c$.
\begin{prop}\label{transitive}
Let $\c$ be an $[n,k,d]$ binary self-dual code. If $\c$ has a 2-transitive automorphism group then
\begin{align*}
Z &= \dfrac{1}{n(n-1)}\dfrac{\partial^2}{\partial x_2\partial x_3}\mathcal{J}_{\c,\c}(x_0,x_1,x_2,x_3)\\
  &= \dfrac{1}{n(n-1)}\dfrac{\partial^2}{\partial x_3\partial x_2}\mathcal{J}_{\c,\c}(x_0,x_1,x_2,x_3). 
\end{align*}
\end{prop}
\begin{proof}
The first part of the proof is taken from \cite{H79}.  We can write the biweight enumerator as $$\mathcal{J}_{\c,\c}(x_0,x_1,x_2,x_3)=\sum A_{i,j,k,l}x_0^ix_1^jx_2^kx_3^l$$ where $A_{i,j,k,l}$ is the number of pairs of codewords with $n_{00}=i,n_{01}=j,n_{10}=k,n_{11}=l$.  For a given coefficient $A_{i,j,k,l}$ and coordinate position $h$, let $N_h(i,j,k,l)$ be the set of all pairs of codewords in $\c$ which contribute to $A_{i,j,k,l}$ and with 01 pattern at $h$.  It follows that $\ds\sum_{h=1}^n|N_h(i,j,k,l)|=jA_{i,j,k,l}$ since any pair in $N_h$ has $j$ positions with the 01 pattern.  Since the automorphism group is transitive then $|N_h(i,j,k,l)|$ is independent of $h$.  Thus, $|N_h(i,j,k,l)|=\frac{j}{n}A_{i,j,k,l}$ and in particular, $|N_2(i,j,k,l)|=\frac{j}{n}A_{i,j,k,l}$.

Let $N'_h(i,j,k,l)$ be the set of all pairs of codewords in $N_2(i,j,k,l)$ with 10 pattern at position $h$.  Using the arguments above and since the automorphism group is 2-transitive, then $|N'_h(i,j,k,l)|$ is independent of $h$ and 	
\begin{align*}
|N'_h(i,j,k,l)|&=\frac{k}{n-1}|N_2(i,j,k,l)|\\
               &=\dfrac{kj}{n(n-1)}A_{i,j,k,l}.
\end{align*}  The proposition now follows. 
\end{proof} 

Since the following examples deal with self-dual codes, we shall remark the following.

\begin{prop}
For a secret sharing scheme with $l=2$ based on a self-dual binary code $\c$, the size of every minimal group in the access structure is even.
\end{prop}

\begin{proof}
 A minimal access group of size $m$ in the access structure corresponds to a pair $(\vec{v}_1,\vec{v}_2)$ of words in $\c^{\perp}=\c$
 such that $\vec{v}_1=(1\;0\; \dots)$ and $\vec{v}_2=(0\;1\;\dots)$ and $m=|(\supp(\vec{v}_1) \cup \supp(\vec{v}_2)) - \{1,2\}|$.
 The latter equals
 $
   \wt(\vec{v}_1)-1 + \wt(\vec{v}_2)-1 -|\supp(\vec{v}_1) \cap \supp(\vec{v}_2)|.
 $
 Since $\c$ is self-dual, the weight of every word in $\c$ is even.
 Moreover, the parity of $|\supp(\vec{v}_1) \cap \supp(\vec{v}_2)|$ equals the inner product of $\vec{v}_1$ with $\vec{v}_2$,
 hence is zero as well.
 Hence $m$ is even.
\end{proof}

\begin{example}
The automorphism group of the $[8,4,4]$ extended Hamming code is 2-transitive and its biweight enumerator is 
\begin{align*}
& \mathcal{J}_{\c,\c}(x_0,x_1,x_2,x_3) = x_3^8+14x_2^4x_3^4+x_2^8+\\
& 14x_3^4x_1^4+14x_2^4x_1^4+x_1^8+168x_0^2x_1^2x_2^2x_3^2+14x_3^4x_0^4+\\
& 14x_2^4x_0^4+14x_1^4x_0^4+x_0^8.
\end{align*}
We obtain $Z=4x_1^3x_2^3 + 12 x_0^2x_1x_2x_3^2$. When $l=2$, the total number of participants is 6.
Since $\frac{3}{2}d^{\perp}-1=5$, we can read off the number of access groups of size 4 as 12.
The only other access group is the one formed by all participants.
\end{example}

\begin{example}
The biweight enumerator of the $[24,12,8]$ Golay code $g_{24}$ was computed in \cite{MMS72} and it is known that the automorphism group of this code is 5-transitive.  Applying the proposition above, we obtain
\begin{align*}
& Z = 6160x_0^{12}x_1^3x_2^3x_3^4 + 22176x_0^{10}x_1^5x_2^5x_3^2+\\
& 7392x_0^{10}x_1^5x_2x_3^6 + 7392x_0^{10}x_1x_2^5x_3^6+\\
& 2640x_0^8x_1^7x_2^7 + 73920x_0^8x_1^7x_2^3x_3^4+\\
& 73920x_0^8x_1^3x_2^7x_3^4 + 36960x_0^8x_1^3x_2^3x_3^8+\\
& 36960x_0^6x_1^9x_2^5x_3^2 + 12320x_0^6x_1^9x_2x_3^6+\\
& 36960x_0^6x_1^5x_2^9x_3^2 + 266112x_0^6x_1^5x_2^5x_3^6+\\
& 7392x_0^6x_1^5x_2x_3^{10} + 12320x_0^6x_1x_2^9x_3^6+\\
& 7392x_0^6x_1x_2^5x_3^{10} + 18480x_0^4x_1^{11}x_2^3x_3^4+\\
& 147840x_0^4x_1^7x_2^7x_3^4 + 73920x_0^4x_1^7x_2^3x_3^8+\\
& 18480x_0^4x_1^3x_2^{11}x_3^4 + 73920x_0^4x_1^3x_2^7x_3^8+\\
& 6160x_0^4x_1^3x_2^3x_3^{12} + 36960x_0^2x_1^9x_2^5x_3^6+\\
& 36960x_0^2x_1^5x_2^9x_3^6 + 22176x_0^2x_1^5x_2^5x_3^{10} + 176x_1^{15}x_2^7+\\
& 672x_1^{11}x_2^{11} + 176x_1^7x_2^{15} + 2640x_1^7x_2^7x_3^8.
\end{align*}
For the secret sharing scheme based on $g_{24}$ with secret lenght $l=2$, the number of groups in the access structure
of size $m=10$ can be read off from $Z$ as $6160$ due to Theorem \ref{main_theorem}, since
$10<\frac{3}{2}d^{\perp}-1=11$.
For every tuple $(\vec{v}_1,\vec{v}_2)$ giving rise to an access group of size $m=12$, we can compute
$P_{D(\vec{v}_1,\vec{v}_2)}(t)$ explicitly, using the information on the pairs of codewords that is given by $Z$.
It turns out that in all the cases, all monomials have degree less than 8, hence due to Corollary \ref{cond_orth},
the number of access groups of size $12$ equals $36960$.
\end{example}

\section{Invariant theory}
Suppose $\c$ is an $[n,k,d]$ binary self-dual code.  We shall apply invariant theory in describing the access structure, similar to what was done in \cite{DMS08}.  We consider the case $l=2$.  Thus, we shall look at the 4-fold joint weight enumerator $\mathcal{J}_{1_{T_1},1_{T_2},\c,\c}(x_a)$ where $a\in\F_2^4$.

If all the codewords of $\c$ have weights divisible by 4 then we have a Type II code.  Otherwise, we have a Type I code.  In \cite{MMS72}, it was shown that the biweight enumerator of a Type I code is invariant under the group $G_1$ generated by all permutation matrices, all 16 matrices diag$(\pm 1,\pm 1,\pm 1,\pm 1)$, and 
$$T_1 = \dfrac{1}{\sqrt{2}}\lt(\begin{array}{cccc}
           1 & 1 & 0 & 0\\ 
           1 & -1 & 0 & 0\\
           0 & 0 & 1 & 1\\
           0 & 0 & 1 & -1\\
         \end{array}\rt).$$
The biweight enumerator of a Type II code is invariant under the group $G_2$ generated by $G_1$ and $T_2 = \text{diag}(1,i,1,i)$ \cite{H79}.

Let $G$ stand for $G_1$ or $G_2$ depending on the type of code we are dealing with.  Following the arguments in \cite{H79} and \cite[Section III]{MMS72}, and using the MacWilliams theorem in \cite{DHO01}, we can verify that $\mathcal{J}_{1_{T_1},1_{T_2},\c,\c}(x_a)$ is left invariant by every element of $G$ acting simultaneously on the following sets of variables:
\begin{eqnarray*}
V_1 &= \{x_0, x_1, x_2, x_3\}\\
V_2 &= \{x_4, x_5, x_6, x_7\}\\
V_3 &= \{x_8, x_9, x_{10}, x_{11}\}\\
V_4 &= \{x_{12}, x_{13}, x_{14}, x_{15}\}.
\end{eqnarray*}
Hence, $\mathcal{J}_{1_{T_1},1_{T_2},\c,\c}(x_a)$ is a simultaneous invariant for the diagonal action of $G$.  As a consequence, we can extend the results in \cite{DMS08} regarding the Molien series.  Note that the exponents of the variables in $V_4$ are always zero, hence we can just consider the remaining three sets.  The vector space of invariants that we are going to use is $\C[x_a]^G_{i,j,k}$ where $x_a\in \F_2^4\setminus V_4$ and $i,j,k$ are the total degrees of the variables in $V_1,V_2,V_3$ respectively.
The corresponding generalized Molien series \cite{St79} is given by
\begin{align*}
\Phi_G(r,s,t) &= \ds\sum_{i=0}^\infty\sum_{j=0}^\infty\sum_{k=0}^\infty\text{dim}(\C[x_a]^G_{i,j,k})\\ 
              &= \dfrac{1}{|G|}\ds\sum_{g\in G} \dfrac{1}{\text{det}(I-rg)\text{det}(I-sg)\text{det}(I-tg)}. 
\end{align*}
Based on the previous section, we are interested in $\text{dim}(\C[x_a]^G_{r,1,1})$.  Its generating function in the variable $r$ is given by $$F_G(r) = \lt.\dfrac{\partial}{\partial s\partial t}\Phi_G(r,s,t)\rt|_{(s,t)=(0,0)}.$$ Using MAGMA \cite{BCP97}, we obtain the following for Type I:
\begin{align*}
F_G(r) &= (r^{20} + r^{16} - 2r^{14} + 2r^{12} + r^{10} + r^{8} - r^6 + 1)\\
       &/(r^{32} - 2r^{30} + 2r^{28} - 4r^{26} + 5r^{24} - 4r^{22} + 6r^{20}\\ 
       &- 6r^{18} + 4r^{16} - 6r^{14} + 6r^{12} - 4r^{10} + 5r^8 - 4r^6\\ 
       &+ 2r^4 - 2r^2 + 1). 
\end{align*}
For Type II we have
\begin{align*}
F_G(r) &= (4r^{62} + 4r^{54} + 5r^{46} + 6r^{38} + 7r^{30}\\ 
       &+ 3r^{22} + 2r^{14} + r^6)\\
       &/(r^{96} - r^{88} - 2r^{72} + 2r^{64} - r^{56} + 2r^{48}\\ 
       &- r^{40} + 2r^{32} - 2r^{24} - r^8 + 1). 
\end{align*}

\section{Conclusion}
We discuss an extension of Massey secret sharing scheme and analyze the access structure using the
dual code and the $g$-fold joint weight enumerator.  It would be worthwhile to replace symmetry properties (group transitivity) by regularity properties (combinatorial designs) in Prop. \ref{transitive}.  Note that for the scheme based on the extended Golay code, we were only able to give a partial description of the access structure.  For future work, we consider the complete description of the access structure.  Another interesting problem is to determine the access structure of schemes based on other families of codes.

\section*{Acknowledgments}
The work of R. dela Cruz was supported by the NTU PhD Research Scholarship and the Merlion PhD Grant.  He would like to thank the hospitality of CNRS I3S Sophia Antipolis and Telecom-ParisTech.  He is also on study leave from the Institute of Mathematics, University of the Philippines Diliman.


\medskip
\begin{thebibliography}{99}
\bibitem{B79} G. Blakley, ``Safeguarding cryptographic keys,'' In \textit{Proc. AFIPS 1979 Natl. Computer Conf.}, N.Y., vol. 48, pp. 313-317, 1979.

\bibitem{BK93} G. Blakley and G. Kabatianskii, ``Linear algebra approach to secret sharing schemes,'' PreProceedings
of Workshop on Information Protection, Moscow, December 1993.

\bibitem{BCP97} W. Bosma, J. Cannon and C. Playoust, ``The Magma Algebra System I: The User Language,'' \textit{Journal of Symbolic Computation}, vol. 24, pp. 235-265, 1997.

\bibitem{CCGHV07} H. Chen, R. Cramer, S. Goldwasser, R. de Haan and V. Vaikuntanathan,`` Secure Computation from Random Error Correcting Codes,'' In \textit{Proceedings of 26th Annual IACR EUROCRYPT}, Barcelona, Spain, Springer Verlag LNCS, vol. 4515, pp. 329-346, May 2007.

\bibitem{DHO01} S. Dougherty, Masaaki Harada and Manabu Oura,`` Note on $g$-fold joint weight enumerators of self-dual codes over $\Z_k$,'' \textit{AAECC}, Vol 11, 437-445, 2001.

\bibitem{DMS08} S. Dougherty, S. Mesnager and P. Sol$\acute{\text{e}}$, ``Secret-sharing schemes based on self-dual codes,'' In \textit{Proceedings of IEEE Information Theory Workshop, ITW 2008}, Porto, Portugal.

\bibitem{H79} W. Huffman, ``The Biweight Enumerator of Self-Orthogonal Binary Codes,'' \textit{Discrete Mathematics}, vol. 26, pp. 129-143, 1979.

\bibitem{MMS72} F. MacWilliams, C. Mallows and N. Sloane, ``Generalizations of Gleason's Theorem on Weight Enumerators of Self-Dual Codes,'' \textit{IEEE Transactions Information Theory}, vol. 18, pp. 794-805, 1972.

\bibitem{MS81} R. McEliece and D. Sarwate, ``On Sharing Secrets and Reed-Solomon Codes,'' \textit{Communications of the ACM}, vol. 24, pp. 583-584, 1981.

\bibitem{M95}
J. L. Massey, ``Some applications of coding theory in cryptography,'' in P.G Farrell (ed.), \textit{Codes and Ciphers, Cryptography and Coding IV}, Formara Lt, Esses, England, pp. 33-47, 1995.

\bibitem{S79}
A. Shamir, ``How to share a secret," \textit{Comm. ACM}, vol. 22, pp. 612-613, November 1979.

\bibitem{St79} R. Stanley, ``Invariants of Finite Groups and their Applications to Combinatorics,'' \textit{Bull. AMS} vol. 3, pp. 475-497, 1979.

\bibitem{St92} D. Stinson, ``An explication of secret sharing schemes,'' \textit{Designs, Codes and Cryptography}, vol. 2, no. 4, pp. 357-390, 1992.
\end{thebibliography}
\end{document}